\documentclass[twocolumn,pra,floatfix,amsmath,superscriptaddress]{revtex4-1}
\usepackage[latin9]{inputenc}
\usepackage{graphicx}
\usepackage[hypertex]{hyperref}
\usepackage{bm,bbm}
\usepackage{amsmath}
\usepackage{amsfonts}
\usepackage{amsthm}
\usepackage{amssymb}
\usepackage{mathrsfs}
\usepackage{enumerate}
\usepackage{subfigure}
\usepackage{array}
\usepackage{color}
\makeatletter

\usepackage{pxfonts}

\newtheorem{theorem}{Theorem}
\newtheorem{lemma}{Lemma}
\newtheorem{definition}{Definition}
\newtheorem{corollary}{Corollary}

\begin{document}

\title{Anomalies of weight-based coherence measure and mixed maximally coherent states}

\author{Yao Yao}
\email{yaoyao@mtrc.ac.cn}
\affiliation{Microsystems and Terahertz Research Center, China Academy of Engineering Physics, Chengdu Sichuan 610200, China}
\affiliation{Institute of Electronic Engineering, China Academy of Engineering Physics, Mianyang Sichuan 621999, China}

\author{Dong Li}
\affiliation{Microsystems and Terahertz Research Center, China Academy of Engineering Physics, Chengdu Sichuan 610200, China}
\affiliation{Institute of Electronic Engineering, China Academy of Engineering Physics, Mianyang Sichuan 621999, China}

\author{C. P. Sun}
\affiliation{Graduate School of China Academy of Engineering Physics, Beijing 100193, China}
\affiliation{Beijing Computational Science Research Center, Beijing 100193, China}

\date{\today}

\begin{abstract}
As an analogy of \textit{best separable approximation} (BSA) in the framework of entanglement theory, here we concentrate on the notion of
\textit{best incoherent approximation}, with application to characterizing and quantifying quantum coherence.
From both analytical and numerical perspectives, we have demonstrated that the weight-based coherence measure displays some \textit{unusual} properties,
in sharp contrast to other popular coherence quantifiers. First, by deriving a closed formula for qubit states,
we have showed the weight-based coherence measure exhibits a rich (geometrical) structure even in this simplest case.
Second, we have identified the existence of \textit{mixed maximally coherent states} (MMCS) with respect to this coherence measure
and discussed the characteristic feature of MMCS in high-dimensional Hilbert spaces.
Especially, we present several important families of MMCS by gaining insights from the numerical simulations.
Moreover, it is pointed out that some considerations in this work can be generalized to general convex resource theories and
a numerical method of improving the computational efficiency for finding the BSA is also discussed.
\end{abstract}

\maketitle
\section{INTRODUCTION}
Recently, the characterization and quantification of quantum coherence remains to be one of the attractive subjects in the field of
quantum information theory, not only for its fundamental implications, but also for practical applications \cite{Streltsov2017,Hu2018}.
It is worth noting that quantum resource theory (QRT) \textit{per se} has also attracted great attention due to its successful application in this topic
\cite{Baumgratz2014,Chitambar2019}. Within the framework of QRT, a plethora of coherence monotones and measures
has been proposed, such as the relative entropy of coherence \cite{Baumgratz2014}, the $l_1$ norm of coherence \cite{Baumgratz2014},
the coherence of formation \cite{Aberg2006,Yuan2015,Winter2016}, the geometric measure of coherence \cite{Streltsov2015} and
the robustness of coherence (ROC) \cite{Napoli2016,Piani2016}.

However, we also notice that the majority of these popular coherence measures are not originally operational defined
with the exception of ROC, which quantifies the minimal ``noise'' or ``mixing'' required to destroy all the coherence
contained in a quantum state \cite{Napoli2016,Piani2016}. In fact, there exists another coherence measure also manifesting
itself as an inherently operational definition, i.e., the weight-based coherence measure, which quantifies the minimal coherence
resource needed to prepare or construct a given state \cite{Bu2018}.

In a specific convex resource theory, the idea of weight-based measure originates from a simple fact that for any given state $\rho$
there always exist convex decompositions such as
\begin{align}
\rho=\lambda\rho_f+(1-\lambda)\rho_r.
\label{decomposition}
\end{align}
Here the \textit{resource} under investigation can be some physical property of quantum states or phenomenon that emerges from the principles of quantum mechanics.
$\rho_f$ belongs to the (convex) set of \textit{free} states, while $\rho_r$ denotes a more resourceful state. When optimizing over
all allowed free states, we will find the maximal weight $\lambda^\star$ according to the decomposition (\ref{decomposition}) and naturally the
weight-based resource measure can be defined as $1-\lambda^\star$.

Actually, the essence of weight-based measure can date back to the Elitzur-Popescu-Rohrlich (EPR2) approach for quantifying nonlocality
of joint probability distributions \cite{Elitzur1992,Barrett2006,Brunner2011}. Meanwhile, the word ``weight'' is also dubbed as ``part'', ``content'', ``cost''
or ``fraction'' in different scenarios and such a line of thought has also been employed to measure quantum entanglement
\cite{Lewenstein1998,Kraus2000,Englert2000,Karnas2001,Wellens2001,Akhtarshenas2003,Akhtarshenas2004,Jafarizadeh2004, Thiang2009,Thiang2010,Quesada2014,Akulin2015,Gabdulin2019},
steering \cite{Skrzypczyk2014,Gallego2015,Das2018}, contextuality \cite{Amselem2012,Grudka2014,Horodecki2015,Abramsky2017},
measurement informativeness \cite{Ducuara2019a} or even arbitrary resources \cite{Ducuara2019b,Uola2019}.

Within the context of convex QRT, although the weight-based resource measure satisfies desirable properties, such as faithfulness, monotonicity and convexity,
it exhibits \textit{unusual} features, in sharp contrast to other popular quantifiers. In this work, we concentrate on the weight-based coherence measure,
or for simplicity we call it \textit{coherence weight}. For instance, the coherence weight is a coarse-grained measure for all pure coherent states,
which means that the coherence weight of any pure coherent state is the same, that is, the maximum value $1$. Indeed, this phenomenon has also been
mentioned for entanglement weight \cite{Cavalcanti2005} and steering weight \cite{Skrzypczyk2014}. Here we explore another two aspects of
coherence weight: (i) by presenting a closed-form formula of coherence weight for single qubit states, we illustrate that the evaluation of coherence weight
depends on the relationship between the absolute values of diagonal and off-diagonal entries (in the incoherent basis). This intriguing fact
is quite remarkable comparing with single-letter formulas of other coherence measures \cite{Baumgratz2014,Winter2016,Streltsov2015,Napoli2016};
(ii) under the framework of QRT, almost all valid coherence measure only assign the maximal value to the maximally coherent \textit{pure} states \cite{Peng2016}.
However, here we propose the notion of \textit{mixed maximally coherent states} (MMCS) with respect to the weight-based coherence measure.
It is worthy noting that such a similar phenomenon, i.e., the existence of mixed maximally steerable states, was demonstrated in quantifying
quantum steering \cite{Skrzypczyk2014}.

The rest of this paper is organized as follows.
In Sec. \ref{sec2}, we review the definition of the weight-based resource measure and discuss its special properties,
especially from the geometric viewpoint.
In Sec. \ref{sec3}, we offer a closed formula of coherence weight for single qubit states, and consequently,
an exhaustive investigation of single-qubit case is put forward by gaining insights from numerical simulations.
In Sec. \ref{sec4}, we propose a definition of mixed maximally coherent states according to the weight-based coherence measure.
Moreover, we provide a detailed numerical analysis of coherence weight for high-dimensional Hilbert spaces and
several important families of MMCS are confirmed.
Discussions and final remarks are given in Sec. \ref{sec5} and several open questions are raised for future research.

\section{Best free approximation and weight-based coherence measure}\label{sec2}
Before focusing on the resource theory of quantum coherence, we begin with the notion of
\textit{best free approximation} (BFA) in a general convex resource theory,
generalising the concept of best separable approximation (BSA) \cite{Regula2018}.
In a $d$-dimensional Hilbert space, let $\mathcal{D}$ and $\mathcal{F}$, respectively,
denote the set of density matrices and the convex set of free states. Hence the BFA of a given state $\rho$
can be defined through an optimization over convex decompositions
\begin{align}
\textrm{BFA}(\rho)&=\min_{\rho_r \in {\mathcal{D}}}\left\{1-\lambda\mid\rho=\lambda\rho_f+(1-\lambda)\rho_r,\rho_f\in\mathcal{F}\right\},\label{BFA}\\
&=\min_{\rho_r \in {\mathcal{D}}}\left\{1-\lambda\mid\rho\succeq \lambda\rho_f,\rho_f\in\mathcal{F}\right\},\label{BFA1}
\end{align}
where the matrix inequality $A\succeq B$ means that $A-B$ is positive semidefinite.
To discuss the properties of BFA, we would like to mention another prominent resource quantifier, that is,
the generalized robustness measure, which is a \textit{dual} quantity to the BFA in some sense \cite{Brandao2005}
\begin{align}
\mathcal{R}(\rho)&=\min_{\tau \in {\mathcal{D}}} \left\{ s\geq 0\ \Big\vert\ \frac{\rho + s\ \tau}{1+s} =: \rho_f \in \mathcal{F}\right\},\label{robustness}\\
&=\min_{\tau \in {\mathcal{D}}}\left\{s\mid\rho\preceq (1+s)\rho_f,\rho_f\in\mathcal{F}\right\}.
\end{align}

Within the framework of QRT, it is generally known that the generalized robustness is a valid resource monotone satisfying
the following axiomatic criteria \cite{Chitambar2019,Regula2018}

\begin{enumerate}[({C}1) ]
\item \textit{Faithfulness}: $\mathcal{R}(\rho)=0$ if and only if $\rho\in\mathcal{F}$.
\item \textit{Convexity}: $\mathcal{R}(\sum_ip_i\rho_i)\leq\sum_ip_i\mathcal{R}(\rho_i)$ for $\rho_i\in\mathcal{D}$, $p_i\geq0$, $\sum_ip_i=1$.
\item \textit{Strong monotonicity}: $\mathcal{R}(\rho)\geq\sum_i\textrm{Tr}[\Theta_i(\rho)]\mathcal{R}[\frac{\Theta_i(\rho)}{\textrm{Tr}[\Theta_i(\rho)]}]$,
where the instrument $\{\Theta_i\}$ is a collection of resource non-generating subchannels, i.e.,
$\Theta_i(\sigma)/\textrm{Tr}[\Theta_i(\rho)]\in\mathcal{F}$ for any $\sigma\in\mathcal{F}$
and $\sum_i\Theta_i$ constitutes a completely positive trace-preserving map.
\end{enumerate}

In fact, the weight-based resource measure is also a sound quantifier fulfilling properties (C1)-(C3):

\begin{lemma}
The BFA($\rho$) is a faithful, convex, and strong monotonic measure in any convex resource theory.
\label{lemma1}
\end{lemma}
\begin{proof}
If we denote by $\tau^\star$ and $\rho_f^\star$ the optimal states achieving the minimum in Eq. (\ref{robustness}),
then $\rho$ can be written as a \textit{pseudomixture} \cite{Sanpera1998}
\begin{align}
\rho=[1+\mathcal{R}(\rho)]\rho_f^\star-\mathcal{R}(\rho)\tau^\star.
\end{align}
Comparing with the convex mixture in the definition of Eq. (\ref{BFA}), we are aware of a crucial fact that any procedure or technic used in proving $\mathcal{R}(\rho)$
to satisfy (C1)-(C3) can be applied to the BFA in the same manner, owing to the linearity of quantum (sub)channels and the convexity
of the set of free states (for instance, one can readily mimic the proof presented in Ref. \cite{Napoli2016}).
\end{proof}

\begin{figure}[htbp]
\begin{center}
\includegraphics[width=0.35\textwidth]{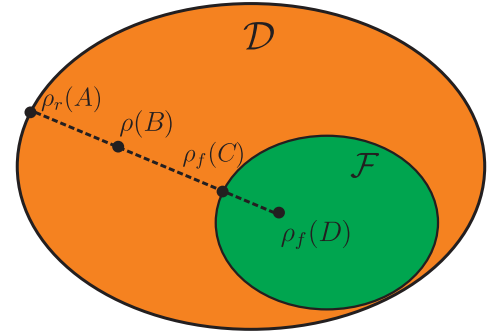}
\end{center}
\caption{(Color online) The geometric interpretation of BFA. The optimal states $\rho_f^\star$ and $\rho_r^\star$ achieving the minimum in Eq. (\ref{BFA})
are on the boundaries of $\mathcal{F}$ and $\mathcal{D}$, respectively. See Appendix \ref{A1} for more details.
}\label{fig1}
\end{figure}

On the other hand, similar to the arguments for the robustness of entanglement \cite{Du2000}, we can give an explicit geometric interpretation
of the BFA:
\begin{lemma}
For any convex resource theory and any (mixed) resourceful state $\rho$, the optimal states $\rho_f^\star$ and $\rho_r^\star$ achieving the minimum in Eq. (\ref{BFA})
are on the boundaries of $\mathcal{F}$ and $\mathcal{D}$, respectively.
\label{lemma2}
\end{lemma}
\begin{proof}
See Appendix \ref{A1}.
\end{proof}

Moreover, note that the BFA is not extensive, i.e., does not scale with the dimension of the state space.
Therefore, when we compare it with other popular resource measures, a proper normalization is necessary:
\begin{lemma}
For any convex resource theory and any other convex resource monotone $\mathcal{X}(\rho)$, the normalized version of $\mathcal{X}(\rho)$ is upper bounded by $BFA(\rho)$, i.e.,
$\overline{\mathcal{X}}(\rho)\leq BFA(\rho)$, where $\overline{\mathcal{X}}(\rho)=\mathcal{X}(\rho)/\mathcal{X}_d$ and $\mathcal{X}_d$ is the maximal value associated with
$d$-dimensional Hilbert space.
\label{lemma3}
\end{lemma}
\begin{proof}
If we denote by $\rho_f^\star$ and $\rho_r^\star$ the optimal states achieving the minimum in Eq. (\ref{BFA}),
that is, $\rho=[1-BFA(\rho)]\rho_f^\star+BFA(\rho)\rho_r^\star$, then from the convexity of $\mathcal{X}(\rho)$ we have
\begin{align}
\mathcal{X}(\rho)&\leq[1-BFA(\rho)]\mathcal{X}(\rho_f^\star)+BFA(\rho)\mathcal{X}(\rho_r^\star)\nonumber\\
&=BFA(\rho)\mathcal{X}(\rho_r^\star)\leq BFA(\rho)\mathcal{X}_d.
\end{align}
After the normalization, the proof is complete.
\end{proof}

Apart form the above general results, we should take a closer look at the resource theory of quantum coherence \cite{Baumgratz2014}.
In this framework, the free states are those diagonal in a prefixed orthogonal basis, i.e., the incoherent basis $\{|i\rangle\}_{i=0}^{d-1}$.
The free operations are usually chosen to the so-called incoherent operations, which admit a Kraus decomposition
$\Theta(\rho)=\sum_iK_i\rho K_i^\dagger$ such that every Kraus operator is required to fulfill
$K_i\rho K_i^\dagger/\textrm{Tr}(K_i\rho K_i^\dagger)\in\mathcal{F}$ for all $\rho\in\mathcal{F}$ \cite{Baumgratz2014,Yao2015}.
However, from Lemma \ref{lemma1}, it is evident that the coherent weight (i.e., best incoherent approximation) is a MIO-monotone, where the abbreviation MIO
stands for ``maximal incoherent operations'', which is recognized as the largest class of incoherent operations
and is equivalent to the definition of resource non-generating channels in Lemma \ref{lemma1} \cite{Aberg2006,Chitambar2016}.

Moreover, the coherence weight can be recast into a simple semidefinite program (SDP) \cite{Bu2018}
\begin{align}
C_w(\rho)=\max\left\{\textrm{Tr}(\rho\omega)\ | \ \Delta\omega\preceq0,\omega\preceq\openone\right\}.\label{SDP}
\end{align}
Nevertheless, from the Eq. (\ref{BFA1}) and the definition of incoherent states, $C_w(\rho)$ can be
directly expressed as
\begin{align}
C_w(\rho)=\min\left\{1-\sum_i\lambda_i\ \Big\vert\ \rho-\sum_i\lambda_i|i\rangle\langle i|\succeq0\right\}.\label{SDP1}
\end{align}
Such an alternative SDP has an advantage over Eq. (\ref{SDP}) in that it can indicate the closest incoherent with respect to $\rho$.
More precisely, we can make use of the {\sc cvx} package to evaluate the coherence weight and simultaneously
obtain the accurate values of $\lambda_i$ \cite{Grant2020}.

Finally, we also notice that for qubit states the $l_1$ norm of coherence, the robustness of coherence,
the coherence of formation, the geometric measure of coherence, etc. are all monotonic functions of the absolute value
of the off-diagonal element $|\rho_{01}|=|\langle 0|\rho|1\rangle|$. In the next section, we show
that it is not the case for coherent weight, which depends on the specific form of the given qubit state
and thus exhibits a richer structure.

\section{$C_w$ for single qubit states}\label{sec3}
In this section, we present the analytical formula of coherence weight for single qubit states.
Before we begin, we recall some important definitions and results initially associated with the notion of BSA
(or so called Lewenstein-Sanpera decomposition) \cite{Lewenstein1998,Karnas2001}.

\begin{definition}
A non-negative parameter $\Lambda$ is called maximal with respect to a density matrix $\rho$ and the projection
operator $P=|\psi\rangle\langle\psi|$ iff $\rho-\Lambda P\succeq0$, and for every $\epsilon\geq0$,
the matrix $\rho-(\Lambda+\epsilon)P$ is not positive definite.
\label{definition1}
\end{definition}
\begin{definition}
A pair of non-negative parameters $(\Lambda_1,\Lambda_2)$ is called maximal with respect to $\rho$ and a pair of projection operators
$P_1=|\psi_1\rangle\langle\psi_1|$, $P_2=|\psi_2\rangle\langle\psi_2|$ iff $\rho-\Lambda_1P_1-\Lambda_2P_2\succeq0$,
$\Lambda_1$ is maximal with respect to $\rho-\Lambda_2P_2$, $\Lambda_2$ is maximal with respect to $\rho-\Lambda_1P_1$,
and the sum $\Lambda_1+\Lambda_2$ is maximal.
\end{definition}
\begin{lemma}
A pair $(\Lambda_1,\Lambda_2)$ is maximal with respect to $\rho$ and a pair of projectors $(P_1,P_2)$ iff

(a) if $|\psi_1\rangle,|\psi_2\rangle\notin r(\rho)$ [where $r(\rho)$ denotes the range of $\rho$] then $\Lambda_1=\Lambda_2=0$.

(b) if $|\psi_1\rangle\notin r(\rho)$ while $|\psi_2\rangle\in r(\rho)$ then $\Lambda_1=0$, $\Lambda_2=\langle\psi_2|\rho^{-1}|\psi_2\rangle^{-1}$.

(c) if $|\psi_1\rangle,|\psi_2\rangle\in r(\rho)$ and $\langle\psi_1|\rho^{-1}|\psi_2\rangle=0$ then $\Lambda_i=\langle\psi_i|\rho^{-1}|\psi_i\rangle^{-1}$, $i=1,2$.

(d) if $|\psi_1\rangle,|\psi_2\rangle\in r(\rho)$ and $\langle\psi_1|\rho^{-1}|\psi_1\rangle, \langle\psi_2|\rho^{-1}|\psi_2\rangle\geq|\langle\psi_1|\rho^{-1}|\psi_2\rangle|\neq0$,
then
\begin{align}
\Lambda_1=\left(\langle\psi_2|\rho^{-1}|\psi_2\rangle-|\langle\psi_1|\rho^{-1}|\psi_2\rangle|\right)/D,\nonumber\\
\Lambda_2=\left(\langle\psi_1|\rho^{-1}|\psi_1\rangle-|\langle\psi_1|\rho^{-1}|\psi_2\rangle|\right)/D,\nonumber
\end{align}
where $D=\langle\psi_1|\rho^{-1}|\psi_1\rangle\langle\psi_2|\rho^{-1}|\psi_2\rangle-|\langle\psi_1|\rho^{-1}|\psi_2\rangle|^2$.

(e) if $|\psi_1\rangle,|\psi_2\rangle\in r(\rho)$ and $\langle\psi_1|\rho^{-1}|\psi_1\rangle\geq|\langle\psi_1|\rho^{-1}|\psi_2\rangle|\geq\langle\psi_2|\rho^{-1}|\psi_2\rangle$,
then $\Lambda_1=0$, $\Lambda_2=\langle\psi_2|\rho^{-1}|\psi_2\rangle^{-1}$.
\label{lemma4}
\end{lemma}

It is worth emphasizing that small typo mistakes in Ref. \cite{Lewenstein1998,Karnas2001} have been corrected here. Besides the original proof of Lemma \ref{lemma4} in Ref. \cite{Lewenstein1998},
we refer the readers to an alternative derivation via SDP method \cite{Jafarizadeh2005}. A key observation in proving Lemma \ref{lemma4} lies in that
although within the context of BSA the projectors are all pure product states such as $|\psi\rangle=|e\rangle\otimes|f\rangle$, actually this point has
not been taken into account in the proof. Obviously, the above lemma can also been exploited in single-partite system through Eq. (\ref{SDP1}),
where for qubit states the projectors are pure incoherent basis states $\{|0\rangle\langle0|,|1\rangle\langle1|\}$.
Now we have the toolkit to present the first main result of this work.

\begin{theorem}
For an arbitrary qubit state $\rho$, the coherence weight can be evaluated as
\begin{align}
C_w(\rho)=
\left\{\begin{array}{cc}
2|\rho_{01}|, &   \rho_{00},\rho_{11}\geq|\rho_{01}|\\
1-\frac{\det\rho}{\min\{\rho_{00},\rho_{11}\}}, &   otherwise
\end{array}\right.
\end{align}
where $\rho_{ij}=\langle i|\rho|j\rangle$ are elements of $\rho$ in the incoherent basis.
\end{theorem}
\begin{proof}
For convenience, we also adopt the Bloch representation of $\rho$, that is
\begin{align}
\rho_{u}=\frac{1}{2}(1+\vec{u}\cdot\vec{\sigma}),\label{Bloch}
\end{align}
where $\vec{u}=(u_1,u_2,u_3)$ denotes the Bloch vector and $\vec{\sigma}=(\sigma_1,\sigma_2,\sigma_3)$ are
standard Pauli matrices. Using Eq. (\ref{Bloch}), one can easily verify that the inverse matrix of $\rho$ can be expressed as
\begin{align}
\rho_{u}^{-1}=\frac{4}{1-|\vec{u}|^2}\rho_{-u}=\frac{1}{\det\rho}\rho_{-u}.
\end{align}
Therefore, up to the factor $1/\det\rho$, the elements of $\rho_{u}^{-1}$ in the incoherent basis
are just that of $\rho_{-u}$. Meanwhile, note that $\langle i|\rho_{-u}|i\rangle=\rho_{jj}$ for $i\neq j=0,1$.

To utilize Lemma \ref{lemma4}, we first assume $\rho_{00},\rho_{11}\geq|\rho_{01}|$.
In this case, we can directly employ item (d) of Lemma \ref{lemma4}
\begin{align}
C_w(\rho)&=1-\frac{\langle0|\rho^{-1}|0\rangle+\langle1|\rho^{-1}|1\rangle-2|\langle0|\rho^{-1}|1\rangle|}
{\langle0|\rho^{-1}|0\rangle\langle1|\rho^{-1}|1\rangle-|\langle0|\rho^{-1}|1\rangle|^2}\nonumber\\
&=1-\frac{\rho_{11}+\rho_{00}-2|\rho_{01}|}{\rho_{11}\rho_{00}-|\rho_{01}|^2}\det\rho\nonumber\\
&=2|\rho_{01}|,
\end{align}
where the relation $\rho_{00}\rho_{11}-|\rho_{01}|^2=\det\rho$ is used. On the other hand, if
the value of $|\rho_{01}|$ lies between $\rho_{00}$ and $\rho_{11}$, item (e) of Lemma \ref{lemma4}
can be applied
\begin{align}
C_w(\rho)=1-\frac{\det\rho}{\min\{\rho_{00},\rho_{11}\}}.
\end{align}
Finally, notice that the combined formula is also valid for pure or incoherent qubit states.
\end{proof}

To gain a deeper insight into $C_w(\rho)$, we perform a numerical simulation for $10^5$ randomly generated qubit states (see Fig. \ref{fig2}).
The following two corollaries encapsulate the intriguing observations and analytical proofs.

\begin{corollary}
For any qubit state $\rho$, there exists a trade-off relationship between the coherence weight $C_w(\rho)$
and the mixedness $M(\rho)$
\begin{align}
C_w(\rho)^2+M(\rho)\leq1,\label{mixedness}
\end{align}
where the mixedness $M(\rho)=2(1-Tr\rho^2)$ is characterized by the linear entropy of $\rho$ \cite{Singh2015}.
\label{corollary1}
\end{corollary}

\begin{figure}[htbp]
\begin{center}
\includegraphics[width=0.48\textwidth]{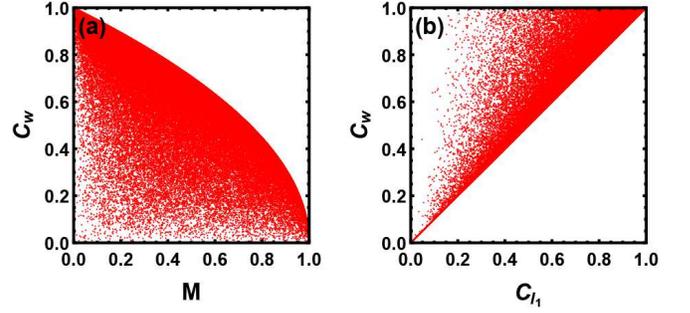}
\end{center}
\caption{(Color online) The numerical simulation (e.g., $10^5$ randomly generated qubit states) shows that (a)
there exists a trade-off relationship between $C_w(\rho)$ and the mixedness $M(\rho)$; (b) $C_w(\rho)$ is always
larger than or equal to the $l_1$ norm of coherence $C_{l_1}(\rho)$ for qubit states.
}\label{fig2}
\end{figure}

\begin{proof}
In the Bloch representation, this inequality is equivalent to $C_w(\rho)\leq|\vec{u}|$, i.e., $C_w(\rho)$
should be less than or equal to the length of the Bloch vector of $\rho$. If $\rho_{00},\rho_{11}\geq|\rho_{01}|$,
this inequality is true since in this case $C_w(\rho)=2|\rho_{01}|=(u_1^2+u_2^2)^{1/2}\leq|\vec{u}|$.
Therefore, without loss of generality, we can assume $\rho_{00}\geq|\rho_{01}|\geq\rho_{11}$, that is, $1+u_3\geq(u_1^2+u_2^2)^{1/2}\geq1-u_3$,
and the following equivalence relations hold
\begin{align}
C_w(\rho)&=1-\frac{2}{1-u_3}\frac{1-|\vec{u}|^2}{4}\leq|\vec{u}|\nonumber\\
&\Leftrightarrow 1-|\vec{u}|\leq\frac{1-|\vec{u}|^2}{2(1-u_3)}\nonumber\\
&\Leftrightarrow 1-2u_3\leq|\vec{u}|\nonumber
\end{align}
The last inequality is valid since $1-2u_3\leq1-u_3\leq(u_1^2+u_2^2)^{1/2}\leq|\vec{u}|$ by assumption.
\end{proof}

In fact, the inequality Eq. (\ref{mixedness}) is stronger than the inequality $C_{l_1}(\rho)^2+M(\rho)\leq1$ proved
in Ref. \cite{Singh2015} due to the fact $C_w(\rho)\geq C_{l_1}(\rho)$ for qubit states (see Lemma \ref{lemma3}).
In the following corollary, a more detailed proof is given starting from the analytical formula of $C_w(\rho)$ and as a byproduct
we can show that the volume of the states ``on the line'' $C_w(\rho)=C_{l_1}(\rho)=2|\rho_{01}|$ is precisely
equal to that ``above this line'' (see Fig. \ref{fig2} (b)).

\begin{corollary}
For any qubit state $\rho$, $C_w(\rho)\geq C_{l_1}(\rho)$ holds. Further we have
\begin{align}
\frac{\nu\left[C_w(\rho)>C_{l_1}(\rho)\right]}{\nu\left[C_w(\rho)=C_{l_1}(\rho)\right]}=1,
\end{align}
where $\nu[\bullet]$ is the volume of states satisfying the corresponding condition.
\label{corollary2}
\end{corollary}
\begin{proof}
See Appendix \ref{A2}.
\end{proof}

To demonstrate the power of Corollary \ref{corollary2}, we have independently generated two sets of
random qubit states (e.g.,  $10^4$ and $10^5$ qubit states respectively) and counted the number of states
``on the line'' (i.e., $C_w(\rho)=C_{l_1}(\rho)$) and ``off the line'' (i.e., $C_w(\rho)>C_{l_1}(\rho)$)
with the accuracy of $10^{-10}$, where the ratios are shown to be $4953:5047$ and $49711:50289$, respectively.

\section{MMCS and Numerical simulations}\label{sec4}
From Corollary \ref{corollary1}, one can infer that only pure qubit states achieve the maximum value of $C_w$.
However, when we consider a Hilbert space of dimension $d\geq3$, the situation is totally different.
Here we introduce the notion of \textit{mixed maximally coherent states} (MMCS), in the spirit of
mixed maximally entangled states and mixed maximally steerable states in entanglement \cite{Cavalcanti2005} and steering theories \cite{Skrzypczyk2014},
respectively. It is important to stress that a seemingly similar (but completely dissimilar) concept called
\textit{maximally coherent mixed states} (MCMS) was proposed to characterize the set of states with maximal coherence for a fixed purity \cite{Singh2015,Yao2016,Streltsov2018}.
First, we provide a necessary condition for the existence of MMCS in $d$-dimensional Hilbert space.

\begin{corollary}
In $d$-dimensional Hilbert space (e.g., $d\geq3$), a MMCS must be rank-deficient, i.e., $rank(\rho_{MMCS})\leq d-1$.
\label{corollary3}
\end{corollary}
\begin{proof}
From the definition of the coherence weight and Lemma \ref{lemma2}, if $C_w(\rho)=1$ then
it must be on the boundary of the set of density matrices, which implies that $\rho$ is of deficient rank.
\end{proof}

In fact, a significant class of MMCS has already been proposed although it has originally been raised for
illustrating the (ir)reversibility of the coherence theory. More precisely,
for a given state the reversibility indicates that its distillable coherence is equal to its coherence cost \cite{Winter2016}.
It is proved that this class of states satisfying reversibility requirement can only be of the following form \cite{Winter2016}
\begin{align}
\rho=\oplus_j p_j|\phi_j\rangle\langle\phi_j|,
\label{reversibility}
\end{align}
where the eigenvectors $\{|\phi_j\rangle\}$ are supported on the orthogonal subspaces spanned by a partition of the incoherent basis.
From Lemma \ref{lemma1}, we have substantially proved that the coherence weight has the property called
\textit{the additivity of coherence for subspace-independent states} \cite{Yu2016}
\begin{align}
C_w(\oplus_j p_j\rho_j)=\sum_jp_jC_w(\rho_j).
\label{additivity}
\end{align}
If we additionally require all coherence ranks $r_c(|\phi_j\rangle)\geq2$, combining Eqs. (\ref{reversibility}) and (\ref{additivity}), we obtain
\begin{align}
C_w(\oplus_j p_j|\phi_j\rangle\langle\phi_j|)=\sum_jp_jC_w(|\phi_j\rangle\langle\phi_j|)=\sum_jp_j=1.
\end{align}
Here the coherence rank $r_c$ of a pure state is defined as its number of nonzero terms in the incoherent basis \cite{Killoran2016}.
In the next theorem, we present another interesting family of the MMCS.
\begin{figure}[htbp]
\begin{center}
\includegraphics[width=0.48\textwidth]{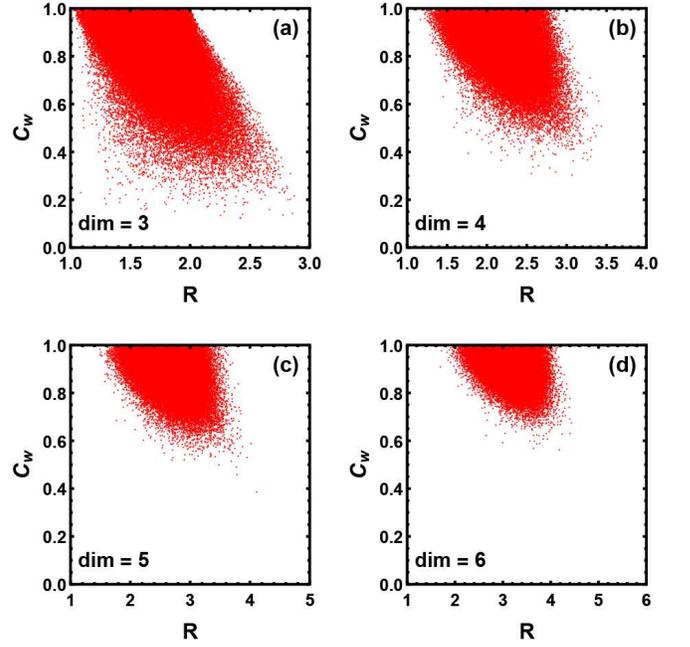}
\end{center}
\caption{(Color online) The coherence weight $C_w(\rho)$ versus the participation ratio $R(\rho)=1/\textrm{Tr}\rho^2$ for
$10^5$ randomly generated states in $d = 3, 4, 5, 6$ dimensions. For qutrit states, it is evident
that the MMCS exist only when $R(\rho)\leq2$.
}\label{fig3}
\end{figure}

\begin{theorem}
If $\rho$ is rank-deficient and an eigenvector $|\psi\rangle$ corresponding to zero eigenvalue
has full coherence rank (i.e., $|\psi\rangle=\sum_{j=0}^{d-1}\psi_j|j\rangle$ with all $\psi_j>0$),
then $\rho$ belongs to the MMCS.
\label{theorem2}
\end{theorem}
\begin{proof}
From Definition \ref{definition1} or directly the SDP form of Eq. (\ref{SDP1}),
we explicitly know that $\rho\in \textrm{MMCS}$ is tantamount to the condition that no incoherent projector $\{|i\rangle\langle i|\}$
can be subtracted from $\rho$ but still maintaining the positivity of the reminder.
Therefore, if we can prove that $\rho$ fulfills this equivalence condition then $\rho\in \textrm{MMCS}$.
Indeed, a square matrix $M$ is said to be positive semidefinite iff $\langle\psi|M|\psi\rangle\geq0$ for any vector $|\psi\rangle$.
Using the assumption of the theorem, we have
\begin{align}
\langle\psi\left| \ \left(\rho-\lambda_i|i\rangle\langle i|\right) \ \right|\psi\rangle=-\lambda_i\left|\langle i|\psi\rangle\right|^2<0,
\end{align}
where $\lambda_i$ is a positive parameter and $|\psi\rangle$ is chosen to be the eigenvector corresponding to zero eigenvalue.
Note that $|\langle i|\psi\rangle|>0$ for all $i=0,\ldots,d-1$ since $|\psi\rangle$ has full coherence rank.
\end{proof}

In Fig. \ref{fig3}, we have plotted the distributions of $C_w(\rho)$ according to the purities of $10^5$ randomly generated states
in $d = 3, 4, 5, 6$ dimensions, where \textit{the participation ratio} $R(\rho)=1/\textrm{Tr}\rho^2$ is introduced \cite{Zyczkowski1998}.
Intriguingly, from our numerical simulations it is found that in every dimension the state with the largest coherent weight (i.e., closest to 1)
falls into the category described in Theorem \ref{theorem2} (see Table \ref{table}).
One step further, in fact we can prove that Theorem \ref{theorem2} is a necessary and sufficient condition for qutrit states and in this case
the MMCS can only exist when the participation ratio $R(\rho)\leq2$ (see Fig. \ref{fig3} (a)).

\begin{table*}
    \renewcommand\arraystretch{1.5}
	\caption{\label{table} For $10^5$ randomly generated states, we have found the state
             with the largest coherence weight $C_w$ in every dimension and also listed
             the coressponding participation ratio $R$, the second smallest eigenvalue $\lambda_{sec}$,
             the minimum eigenvalue $\lambda_{min}$, and the coherence rank $r_c$ of the eigenvector $|\psi\rangle_{min}$  associated with $\lambda_{min}$.}	
	\begin{ruledtabular}
		\begin{tabular}{cccccc}
			$d$  &$C_w$       & $R$   &$\lambda_{sec}$  &$\lambda_{min}$         &$r_c(|\psi\rangle_{min})$      \\ \colrule
			3    &0.999992    &1.59051     &0.246295    &$9.93186\times10^{-7}$                     &3                            \\
			4    &0.999997    &2.30158     &0.129657    &$5.64599\times10^{-8}$                     &4                             \\
			5    &0.999999    &2.61211     &0.060884    &$3.99851\times10^{-8}$                     &5                            \\
			6    &0.999999    &2.91176     &0.020337    &$3.32934\times10^{-8}$                     &6
		\end{tabular}
	\end{ruledtabular}
\end{table*}

\begin{corollary}
A mixed qutrit state $\rho$ is a MMCS if and only if it is of the form
\begin{align}
\rho=\lambda_1|\psi_1\rangle\langle\psi_1|+\lambda_2|\psi_2\rangle\langle\psi_2|,
\end{align}
where $\lambda_1,\lambda_2>0$ with $\lambda_1+\lambda_2=1$, $\langle\psi_1|\psi_2\rangle=0$ and the eigenvector $|\psi_0\rangle$
associated with zero eigenvalue has full coherence rank.
Moreover, the qutrit MMCS can only exist for states with $R(\rho)\leq2$.
\end{corollary}
\begin{proof}
The ``if'' part is obviously valid owing to Theorem \ref{theorem2}. On the other hand, if $C_w(\rho)=1$ for a
\textit{mixed} qutrit state, it can be drawn from Corollary \ref{corollary3} that $\rho$ has two strictly positive eigenvalues.
We consider the spectral decomposition of $\rho$
\begin{align}
\rho=\lambda_1|\psi_1\rangle\langle\psi_1|+\lambda_2|\psi_2\rangle\langle\psi_2|,
\end{align}
where $\lambda_1,\lambda_2>0$ with $\lambda_1+\lambda_2=1$ and $\langle\psi_1|\psi_2\rangle=0$.

Without loss of generality, we assume that $r_c(|\psi_1\rangle)\leq r_c(|\psi_2\rangle)$ and
consider all possible options concerning the coherence ranks of $|\psi_1\rangle$ and $|\psi_2\rangle$:

(i) if $r_c(|\psi_1\rangle)=r_c(|\psi_2\rangle)=1$, then $\rho$ is incoherent which contradicts $C_w(\rho)=1$;

(ii) if $r_c(|\psi_1\rangle)=1,r_c(|\psi_2\rangle)=2$, from Eq. (\ref{additivity}) we have $C_w(\rho)=\lambda_2<1$
which also contradicts $C_w(\rho)=1$;

(iii) if $r_c(|\psi_1\rangle)=r_c(|\psi_2\rangle)=2$, then $\rho$ is reduced to a mixed \textit{qubit} state while
there does not exist MMCS for qubit system;

(iv) if $r_c(|\psi_1\rangle)=2,r_c(|\psi_2\rangle)=3$, a typical example can be expressed as
\begin{align}
|\psi_1\rangle&=\alpha|0\rangle+\beta|1\rangle,\nonumber\\
|\psi_2\rangle&=-\gamma\beta^\ast|0\rangle+\gamma\alpha^\ast|1\rangle+\sqrt{1-\gamma^2}|2\rangle,\nonumber
\end{align}
where $\alpha,\beta$ are non-zero complex numbers and $|\alpha|^2+|\beta|^2=1$, $0<\gamma<1$.
Since $|\psi_0\rangle$ is also orthogonal to $|\psi_1\rangle$, it must have the same form as $|\psi_2\rangle$ with $0\leq\gamma_0\neq\gamma\leq1$
\begin{align}
|\psi_0\rangle=-\gamma_0\beta^\ast|0\rangle+\gamma_0\alpha^\ast|1\rangle+\sqrt{1-\gamma_0^2}|2\rangle.\nonumber
\end{align}
The condition $\langle\psi_0|\psi_2\rangle=0$ is equivalent to
\begin{align}
\gamma\gamma_0+\sqrt{1-\gamma^2}\sqrt{1-\gamma_0^2}=0,\nonumber
\end{align}
which is impossible when $\gamma_0=0$ or $1$ (i.e., $r_c(|\psi_0\rangle)=1$ or $2$). Thus in this case we have $r_c(|\psi_0\rangle)=3$;

(v) if $r_c(|\psi_1\rangle)=r_c(|\psi_2\rangle)=3$, we first consider the probability $r_c(|\psi_0\rangle)=2$.
More precisely, if $r_c(|\psi_0\rangle)=2$, the role of $|\psi_0\rangle$ played in this case
is the same as $|\psi_1\rangle$ in the above case. Typically, we can assume $|\psi_0\rangle=\alpha|0\rangle+\beta|1\rangle$
then $|\psi_1\rangle$ and $|\psi_2\rangle$ can be expressed as ($0<\gamma_1\neq\gamma_2<1$)
\begin{align}
|\psi_1\rangle&=-\gamma_1\beta^\ast|0\rangle+\gamma_1\alpha^\ast|1\rangle+\sqrt{1-\gamma_1^2}|2\rangle,\nonumber\\
|\psi_2\rangle&=-\gamma_2\beta^\ast|0\rangle+\gamma_2\alpha^\ast|1\rangle+\sqrt{1-\gamma_2^2}|2\rangle.\nonumber
\end{align}
Next, for instance we can check the positive semi-definiteness of $\rho-x|2\rangle\langle2|$ for some strictly positive coefficient $x$.
For any vector $|\varphi\rangle$ in the three-dimensional space, $\rho-x|2\rangle\langle2|\succeq0$ is equivalent to the positivity of the following expression
\begin{align}
\langle\varphi\left| \ \left(\rho-x|2\rangle\langle2|\right) \ \right|\varphi\rangle=\langle\varphi|\rho|\varphi\rangle-x|\langle2|\varphi\rangle|^2.
\end{align}
Note that when $|\varphi\rangle$ is chosen to be $|\psi_0\rangle$, $x$ can be an arbitrary positive number.
Therefore, if we choose $x$ to be a constant within the realm of $(0,x^\star]$ where
\begin{align}
x^\star=\min_{|\varphi\rangle\neq|\psi_0\rangle}\frac{\langle\varphi|\rho|\varphi\rangle}{|\langle2|\varphi\rangle|^2}>0,
\label{bound}
\end{align}
then the positive semi-definiteness of $\rho-x|2\rangle\langle2|$ is guaranteed.
Actually, one can further infer that $x$ is upper bounded by
\begin{align}
x\leq\frac{1}{\langle2|\rho^{-1}|2\rangle}<1,
\end{align}
where the upper bound is achieved when (unnormalized) $|\varphi\rangle$ is set to be $\rho^{-1}|2\rangle$ in Eq. (\ref{bound}) \cite{Lewenstein1998}.
This implies that the projector $|2\rangle\langle2|$ can be subtracted from $\rho$ ``by some amount'' but still maintaining the positivity of the reminder,
which contradicts with the assumption $\rho\in MMCS$. Finally it can be concluded that $r_c(|\psi_0\rangle)=3$.

Moreover, the participation ratio $R(\rho)$ is restricted to
\begin{align}
R(\rho)=\frac{1}{\lambda_1^2+\lambda_2^2}\leq\frac{2}{(\lambda_1+\lambda_2)^2}=2.
\end{align}
This result is clearly illustrated in Fig. \ref{fig3} (a).
\end{proof}

\section{DISCUSSION AND CONCLUSION}\label{sec5}
In this work, we begin with the notion of BFA in a general convex resource theory,
which generalizes the concept of BSA \cite{Regula2018}. By presenting three crucial lemmas,
we have exhibited the universal properties of this resource measure.
Concentrating on the specific resource theory of coherence, an analytical formula
of coherence weight has been derived for any qubit state. In fact, the value of
coherent weight for a given qubit state relies on the relationship between the diagonal and off-diagonal elements
in the incoherent basis, which is in sharp contrast with other popular coherence monotones.
Furthermore, as another particular feature of weight-based resource measure,
we have introduced the notion of MMCS, in the spirit of
mixed maximally entangled states \cite{Cavalcanti2005} and mixed maximally steerable states \cite{Skrzypczyk2014}.
Combining with numerical simulations, we have presented two families of MMCS in arbitrary dimension (e.g. $d\geq3$)
and completely characterized the form of qutrit MMCS.

Although recently an operational interpretation has been proposed for the convex weight in general quantum resource theories \cite{Ducuara2019b,Uola2019},
there still exist many open questions, especially focusing on the computability of weight-based quantifiers.
For instance, while the coherence weight can be solved effectively by the SDP method, it is still worth making the effort
to find analytical results for high-dimensional states. Moreover, we notice that at present there is no
universal and efficient algorithm to determine the BSA of arbitrary states, mainly due to the fact that the set of separable states is not a polytope \cite{Ioannou2006}.
However, recently Lu \textit{et al.} established a separability-entanglement classifier
by using an iterative algorithm of convex hull approximation in Ref. \cite{Lu2018}.
Actually, we realize that such a method can also be adopted to compute the entanglement weight and
our numerical simulations show that it is less time-consuming than that of Ref. \cite{Akulin2015,Gabdulin2019}.
Subsequent work is already underway concerning these considerations.

\begin{acknowledgments}
Y.Y. is particularly grateful for helpful discussions with G.H. Dong.
This research is supported by Science Challenge Project (Grant No. TZ2018003) and National Natural Science
Foundation of China (NSFC) (Grants No. 11605166 and No. 61875178).
C.P.S. acknowledges financial support from NSFC (Grant No. 11534002),
NSAF (Grants No. U1930402 and No. U1930403), and National Basic Research Program of China (Grant No. 2016YFA0301201).
\end{acknowledgments}
\appendix

\section{Geometric interpretation of BFA}\label{A1}
The geometric interpretation is already illustrated in Fig. \ref{fig1}, where
$\rho_f\in\mathcal{F}$ and $\rho_r\in\mathcal{D}/\mathcal{F}$ can be viewed as points on
the straight line across the fixed point $\rho$.
Note that any convex decomposition in the definition of Eq. (\ref{BFA}) can be rewritten as
\begin{align}
\rho-\rho_r=\lambda(\rho_f-\rho_r).
\end{align}
By employing a valid distance metric, e.g., the Hilbert-Schmidt norm, we have
\begin{align}
\lambda=\frac{\|\rho-\rho_r\|}{\|\rho_r-\rho_f\|}=\frac{l_{AB}}{l_{AC}},\label{geometry}
\end{align}
where the lengths of the line segments of $\overline{AB}$ and $\overline{AC}$ are denoted by $l_{AB}$ and $l_{AC}$, respectively.

To approach the maximum value of the weight $\lambda$, we can simply take two steps. First, for a given state $\rho$ (i.e., the point $B$), we can fix
$\rho_r$ (i.e., the point $A$) in the interior of $\mathcal{D}$, then from Eq. (\ref{geometry}) it is obvious that $\rho_f$ should
be chosen on the boundary of $\mathcal{F}$ (i.e., the point $C$) in order to minimize the length $l_{AC}$.
In fact, if $\rho_f$ is in the interior region of $\mathcal{F}$ (i.e., the point $D$), apparently we have $l_{AD}>l_{AC}$.
Second, when $\rho_f$ is fixed on the boundary of $\mathcal{F}$, it is clear that
$\lambda$ is a monotonic increasing function of $l_{AB}$ since
\begin{align}
\lambda=\frac{l_{AB}}{l_{AC}}=\frac{l_{AB}}{l_{AB}+l_{BC}}.
\end{align}
Therefore, $\rho_r$ (i.e., the point $A$) should settle on the boundary of $\mathcal{D}$ for maximizing $l_{AB}$.

\section{Proof of corollary \ref{corollary2}}\label{A2}
\begin{figure}[htbp]
\begin{center}
\includegraphics[width=0.48\textwidth]{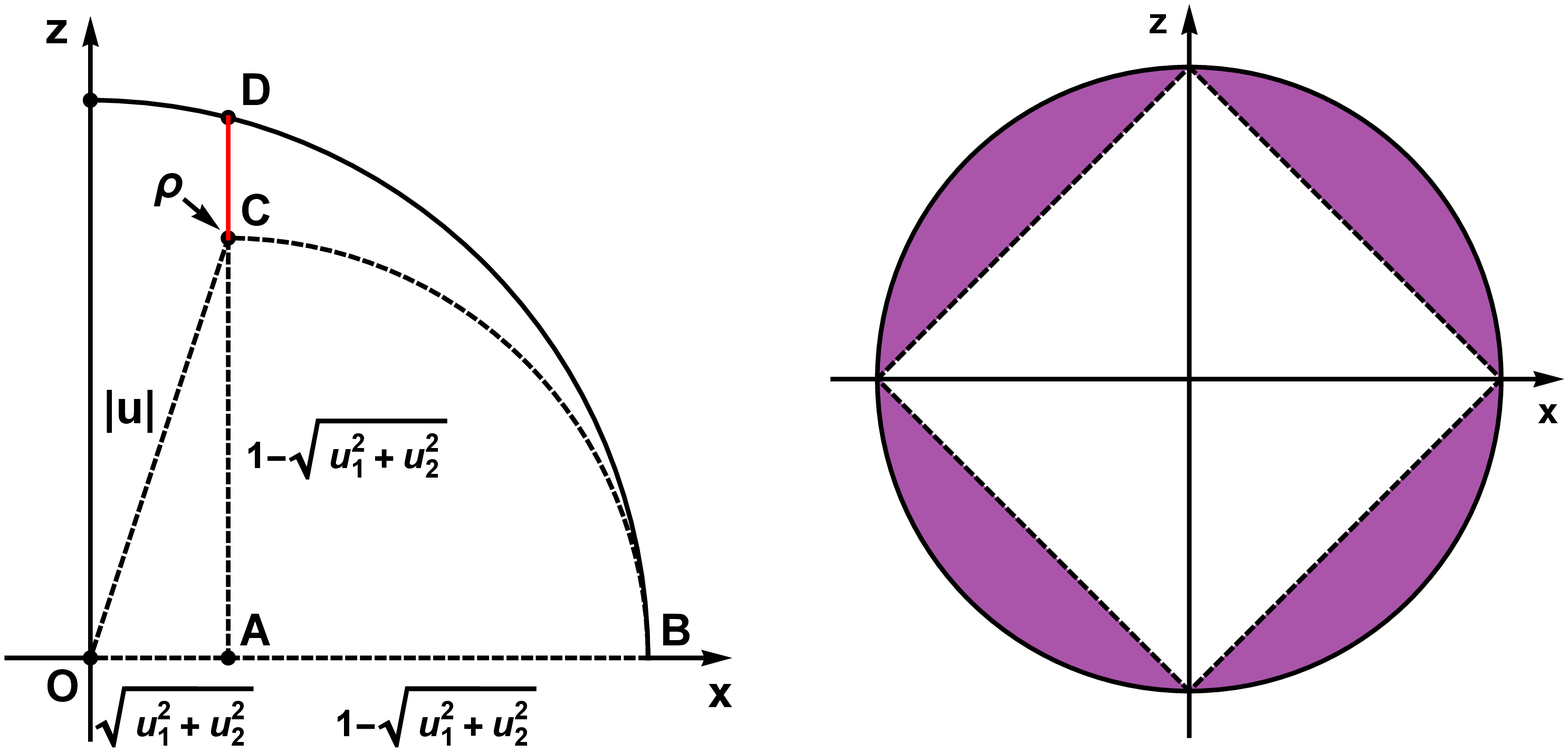}
\end{center}
\caption{(Color online) The geometric interpretation of Corollary \ref{corollary2}. The restriction $u_3\geq1-(u_1^2+u_2^2)^{1/2}$
renders the allowed state $\rho$ on the line segment $\overline{CD}$. See the main text for more descriptions.
}\label{fig4}
\end{figure}
To prove $C_w(\rho)\geq C_{l_1}(\rho)$, we only need to consider the case that
the value of $|\rho_{01}|$ lies between $\rho_{00}$ and $\rho_{11}$.
Due to the perfect rotational symmetry of the Bloch sphere, we focus on a quarter of the unit disk shown in Fig. \ref{fig4}
and assume $\rho_{00}\geq|\rho_{01}|\geq\rho_{11}$, i.e., $u_3\geq1-(u_1^2+u_2^2)^{1/2}$.
It is worth pointing out that the point A is actually the projection of $\rho$ on the $x-y$ plane in the original Bloch sphere.
Thus we have $l_{OA}=C_{l_1}(\rho)=(u_1^2+u_2^2)^{1/2}$ and if we fix the length of $l_{OA}$ then the allowed state $\rho$
with constant $C_{l_1}(\rho)$ can move on the line segment connecting the point $A$ and $D$. However, the condition
$u_3\geq1-(u_1^2+u_2^2)^{1/2}$ results in the fact that $\rho$ is further restricted on the line segment $\overline{CD}$,
where the \textit{critical} state on the point $C$ satisfies the equality $u_3=1-(u_1^2+u_2^2)^{1/2}$.

With this geometric representation in mind, now we deal with the inequality
\begin{align}
C_w(\rho)=1-\frac{1-|\vec{u}|^2}{2(1-u_3)}\geq\sqrt{u_1^2+u_2^2}=2|\rho_{01}|.
\end{align}
If we define the function of $u_3$
\begin{align}
f(u_3)=1-\frac{1-|\vec{u}|^2}{2(1-u_3)}-\sqrt{u_1^2+u_2^2},
\end{align}
then the validity of the inequality is equivalent to $f(u_3)\geq0$ for $1-(u_1^2+u_2^2)^{1/2}\leq u_3\leq[1-(u_1^2+u_2^2)]^{1/2}$.
First, note that when $\rho$ settles on the point $C$ one can easily verify that $f(u_3)=0$.
Moreover, we have
\begin{align}
f'(u_3)&=\frac{u_3}{1-u_3}-\frac{1-u_3^2-a^2}{2(1-u_3)^2},\nonumber\\
f''(u_3)&=\frac{a^2}{(1-u_3)^3}\geq0,\nonumber
\end{align}
where for simplicity we define $a=(u_1^2+u_2^2)^{1/2}$.
Since $f'(u_3=1-a)=0$ and $f''(u_3)\geq0$, $f(u_3)$ is a convex function of $u_3$ for $1-a\leq u_3\leq(1-a^2)^{1/2}$,
and thus the minimum value of $f(u_3)$ is achieved at the boundary $u_3=1-a$. Therefore we have
$f(u_3)\geq f(u_3=1-a)=0$.

On the other hand, when $l_{OA}=(u_1^2+u_2^2)^{1/2}$ varies the line segment $\overline{CD}$ evolves into
the purple shaded regions of the Bloch sphere, implying that the qubit states satisfying $C_w(\rho)>C_{l_1}(\rho)$
occupy such regions (see Fig. \ref{fig4}). By use of the formula of the volume of a cone, we obtain
\begin{align}
\frac{\nu\left[C_w(\rho)>C_{l_1}(\rho)\right]}{\nu\left[C_w(\rho)=C_{l_1}(\rho)\right]}=
\frac{\frac{4}{3}\pi r^3-2\times\frac{1}{3}\pi r^2h}{2\times\frac{1}{3}\pi r^2h}=1,
\end{align}
where the height of the cone $h$ is equal to the radius $r$ of the Bloch sphere.
The proof is complete.


\end{document}